%% file: main.tex
\documentclass[runningheads,a4paper]{llncs}

\usepackage[english]{babel}
\usepackage[T1]{fontenc}
\usepackage{lmodern}
\usepackage{amsmath}
\usepackage{amssymb}
\usepackage{booktabs}
\usepackage[dvipsnames]{xcolor}
\usepackage{graphicx}
\usepackage{placeins}
\usepackage{tikz}
\usepackage{xspace}
\usetikzlibrary{arrows.meta,backgrounds,calc,positioning,quotes,shapes.geometric}
\usepackage{tabularx}
\usepackage{float}
\usepackage[expansion=false]{microtype}
\usepackage[switch,mathlines]{lineno}
\usepackage[
  colorlinks=true,
  allcolors=blue
]{hyperref}
\hypersetup{
    pdftitle={On the Parameterized Complexity of Coloring Discovery},
    pdfauthor={Eric Decker and Sebastian Siebertz},
}
\usepackage[dvipsnames]{xcolor}
\usepackage[nameinlink]{cleveref}
\usepackage[all,defaultlines=3]{nowidow}

\newcommand{\pname}[1]{\textnormal{\textsc{#1}}\xspace}
\newcommand{\CD}{\pname{Coloring Discovery}}
\newcommand{\EC}{\pname{Equitable Coloring}}
\newcommand{\LC}{\pname{List Coloring}}
\newcommand{\FlipCD}{\pname{Flip-CD}}
\newcommand{\SwapCD}{\pname{Swap-CD}}
\newcommand{\SlideCD}{\pname{Slide-CD}}
\newcommand{\VC}{\operatorname{vc}}
\newcommand{\DTC}{\operatorname{dtc}}
\newcommand{\FVS}{\operatorname{fvs}}
\newcommand{\TD}{\operatorname{td}}
\newcommand{\BW}{\operatorname{bw}}
\newcommand{\DDP}{\operatorname{ddp}}
\newcommand{\diam}{\operatorname{diam}}
\newcommand{\Znonneg}{\mathbb{Z}_{\ge 0}}
\newcommand{\cO}{\mathcal{O}}
\newcommand{\FPT}{\mathsf{FPT}}
\newcommand{\Wone}{\mathsf{W[1]}}
\newcommand{\paraNP}{\mathsf{para\text{-}NP}}
\newcommand{\NP}{\mathsf{NP}}
\newcommand{\red}{\textsf{r}}
\newcommand{\green}{\textsf{g}}
\newcommand{\blue}{\textsf{b}}

\tikzset{
  vertex/.style={draw,circle,inner sep=0pt,minimum size=3mm},
  small vertex/.style={draw,circle,inner sep=0pt,minimum size=2mm},
  mwpm node/.style={draw,rectangle,minimum size=5mm,font=\footnotesize},
  text node/.style={anchor=west,font=\small},
  every edge quotes/.style={font=\scriptsize,fill=white,inner sep=0pt}
}

\begin{document}

\title{On the Parameterized Complexity of Coloring Discovery}
\titlerunning{The Parameterized Complexity of Coloring Discovery}

\author{Eric Decker and Sebastian Siebertz}
\authorrunning{E. Decker and S. Siebertz}
\institute{University of Bremen, Bremen, Germany}

\maketitle

\begin{abstract}
Coloring Discovery asks whether a possibly improper initial coloring can be made proper within a prescribed number of allowed changes. 
We study the parameterized complexity of three modification step models that were studied previously in the literature: recoloring one vertex (color flipping), swapping the colors of arbitrary vertices (color swapping), and swapping colors only across an edge (color sliding). For color flipping, we give exact fixed-parameter algorithms for the parameters vertex cover and distance to complete. 
For color swapping, we obtain fixed-parameter tractability for the parameter vertex cover plus the number of colors. 
Our lower bounds show $\Wone$-hardness for treedepth plus feedback vertex set in the color flipping model and for the number of colors plus bandwidth or distance to disjoint paths in the swapping and sliding models.
All three variants remain $\NP$-complete with four colors on graphs of diameter two.
\end{abstract}

\section{Introduction}

Several established algorithmic models study how existing solutions can be
adapt\-ed rather than computed from scratch. In combinatorial reconfiguration,
one is given two feasible solutions and asks whether one can be transformed
into the other by a sequence of small modifications, usually requiring every
intermediate configuration to remain feasible~\cite{bousquet2024survey,vandenheuvel2013complexity}.
Local search explores a prescribed neighborhood of a current solution, aiming to find a solution with improved costs~\cite{aarts1997local,fellows12localsearch}.
Reoptimization starts with a solution to an original instance and studies how
this information can be exploited after the instance changes~\cite{boeckenhauer2018reoptimization,schieber18reopti}.
In recoverable robust optimization, an initial solution is chosen so that,
after an uncertain scenario is revealed, it can be turned into a feasible
solution using only limited recovery actions~\cite{liebchen2009recoverable}.

Many graph coloring applications have this dynamic character. Consider an examination timetable that has already been published. The vertices represent exams, an edge joins two exams that cannot take place simultaneously, and the colors represent time slots. A late change in enrollments may introduce a new conflict between two exams that currently occupy the same slot, so that the existing timetable is no longer feasible. Computing an arbitrary proper coloring from scratch may reschedule a large part of the timetable, although the new conflict could perhaps be resolved by moving only one or two exams. A similar situation occurs in frequency assignment: adding a transmitter, moving an existing one, or increasing its broadcasting range may create interference in an otherwise functioning network. Retuning many devices is costly, and one would prefer to restore feasibility by a small number of controlled changes. In such situations, the initial assignment is part of the problem and cannot simply be discarded~\cite{debiasi2019restoring,fellows2026solution,garnero2018fixing}.

For graph coloring, several concrete modification step models have been studied. 
Color-Fixing asks whether a possibly improper coloring can be made proper by recoloring a bounded number of vertices. Its restriction to a fixed color palette of size $r$ is called $r$-Fix~\cite{debiasi2019restoring,garnero2018fixing}. The $r$-Swap problem instead permits exchanges of the colors of arbitrary pairs of vertices~\cite{debiasi2019restoring}. In both models, the target coloring is not prescribed.

Two fixed-target models are also closely related. Classical coloring reconfiguration starts and ends with prescribed proper colorings and requires every intermediate coloring to remain proper.
Representative results include the boundary between three and four colors~\cite{bonsma2009paths,cereceda2011paths}, bounded-length recoloring~\cite{bonsma2014bounded}, and reconfiguration under adjacent color swaps, which we refer to as color sliding~\cite{fuchs2025swapping}. Token Swapping and Colored Token Swapping use the same local exchange operation as color sliding, but have a prescribed target arrangement: an individual destination for every token in Token Swapping, and a target color for every vertex in Colored Token Swapping~\cite{bonnet2018tokens,yamanaka2015labeled,yamanaka2018colored}.

These lines of work lead to a natural question: starting from the current, possibly infeasible state, can we reach some feasible state by a bounded number of modification steps? Solution discovery offers a general framework for this question~\cite{fellows2026solution}. A solution discovery problem specifies a base problem and a set of permitted modification steps. A problem instance comprises an instance of the base problem, an initial configuration, and a budget. The target configuration is not fixed and the initial and intermediate configurations need not be feasible. The task is to find any feasible configuration that can be reached within the budget. Thus, solution discovery is based on the local modification rules familiar from reconfiguration, but replaces the prescribed target by the objective of repairing the current state. 
Fellows et al. introduced this framework and studied its coloring, vertex-cover, independent-set, and dominating-set variants~\cite{fellows2026solution}. Further work has investigated solution-discovery variants of problems solvable in polynomial time~\cite{grobler2024problemsinp}, kernelization~\cite{grobler2024kernelization}, logical meta-theorems~\cite{bousquet2025metatheorems}, path problems in which feasibility and movement live on different graphs~\cite{vonbergen2026path}, and first-order value objectives~\cite{gerhard2026fovalue}. Saito et al. consider further discovery variants of classical vertex-subset problems~\cite{saito2026solution}.

In this paper we study \CD under the three modification step models described above: color flipping, color swapping, and color sliding. The flip rule coincides with Color-Fixing (and, for a fixed palette of size $r$, with $r$-Fix), while the swap rule coincides with $r$-Swap. We are given a graph $G$, a possibly improper $q$-coloring $\varphi$, and a budget $b$.

We use $\FPT$ and $\NP$ for fixed-parameter tractability and nondeterministic polynomial time, respectively. $\Wone$ and $\paraNP$ denote the standard parameterized complexity classes and the ETH refers to the Exponential Time Hypothesis. The parameters $\VC(G)$, $\DTC(G)$, $\operatorname{tw}(G)$, $\TD(G)$, $\FVS(G)$, $\BW(G)$, $\DDP(G)$, and $\diam(G)$ denote, respectively, the vertex-cover number, distance to complete, treewidth, treedepth, feedback vertex-set number, bandwidth, distance to disjoint paths, and diameter of $G$. Finally, $n=|V(G)|$ and $m=|E(G)|$ denote the number of vertices and edges, $q$ the number of colors, and $b$ the budget. Recall that a problem is fixed-parameter tractable with respect to a parameter $k$ if an encoded instance~$\mathcal I$ can be solved in time $f(k)|\mathcal I|^{\cO(1)}$ for some computable function $f$. Unless $\FPT=\Wone$, a $\Wone$-hard problem is not fixed-parameter tractable.
\medskip

Before the present work, the following results were known.

\begin{itemize}
  \item Garnero et al. prove that the flip variant is polynomial-time solvable
  for $q\le2$ and, for every fixed $q\ge3$, $\NP$-complete even on planar
  bipartite graphs. When $q$ is part of the input, it is $\Wone$-hard
  parameterized by the budget. They give fixed-parameter algorithms
  parameterized by $q+b$ and by $q+\operatorname{tw}(G)$
  ~\cite{garnero2018fixing}.
  \item De Biasi and Lauri prove that, for every fixed $q\ge3$, the swap
  variant is $\NP$-complete even on planar bipartite graphs and $\Wone$-hard
  parameterized by the budget. When $q$ is part of the input, both the flip and
  swap variants are $\Wone$-hard parameterized by treewidth alone
  ~\cite{debiasi2019restoring}.
  \item Fellows et al. introduce the slide variant of \CD and give
  polynomial-time algorithms for the swap and slide variants with two colors.
  They extend the planar-bipartite $\NP$-completeness result to the slide
  variant through a unified proof for all three variants, give fixed-parameter
  algorithms on structurally nowhere dense classes parameterized by $q+b$,
  and prove that, on general graphs, the slide variant is $\Wone$-hard
  parameterized by $q+b$ and by treewidth~\cite{fellows2026solution}.
\end{itemize}

\paragraph{Our results.}
 \Cref{tab:results} summarizes our main new results.

\begin{table}[H]
\caption{Main complexity results proved in this paper. }
\label{tab:results}
\centering
\small
\begin{tabularx}{\linewidth}{
    @{}
    l
    l
    >{\raggedright\arraybackslash}X
    @{}
}
\toprule
Operation & Parameter & Result \\
\midrule
flip
    & $k=\VC(G)$
    & $\FPT$, 
      $2^{\cO(k\log k)}(n+m)$ \\

flip
    & $k=\DTC(G)$
    & $\FPT$, $2^{\cO(k^2)}n^3$ \\

swap
    & $\VC(G)+q$
    & $\FPT$ \\

flip
    & $\TD(G)+\FVS(G)$
    & $\Wone$-hard \\

swap, slide
    & $q+\BW(G)$
    & $\Wone$-hard \\

swap, slide
    & $q+\DDP(G)$
    & $\Wone$-hard\\

all three
    & $q+\diam(G)$
    & $\NP$-complete already on connected graphs with
      $q=4$ and $\diam(G)=2$ \\

all three
    & $k=\VC(G)$
    & no universal $\theta^k n^{\cO(1)}$-time algorithm for all
      fixed $q$, and no $2^{o(k \log k)} n^{\cO(1)}$ algorithm, unless ETH fails \\
\bottomrule
\end{tabularx}
\end{table}

For the flip variant, we give exact fixed-parameter algorithms parameterized by vertex cover and by distance to complete. For $k=\VC(G)$, the running time is $2^{\cO(k\log k)}(n + m)$, and for $k=\DTC(G)$, it is $2^{\cO(k^2)}n^3$. For the swap variant, we prove fixed-parameter tractability parameterized by $\VC(G)+q$. The parameterized complexity of the slide variant for $\VC(G)+q$ remains open.

For the lower bounds, we prove that the flip variant is $\Wone$-hard parameterized by $\TD(G)+\FVS(G)$. Both the swap and slide variants are $\Wone$-hard parameterized by $q+\BW(G)$ and by $q+\DDP(G)$. 
All three variants remain $\NP$-complete on connected graphs of diameter two with four colors. Consequently, they are $\paraNP$-complete parameterized by $q+\diam(G)$, and also by $q$ together with either the radius or the domination number. Finally, unless the Exponential Time Hypothesis fails, none of the three variants admits a constant single-exponential base $\theta$ such that, for all fixed values of $q$, \CD can be solved in time $\cO\big(\theta^{\VC(G)} n^{\cO(1)}\big)$.

Taken together, our results augment the parameterized-complexity landscape of Coloring Discovery and settle the complexity of the three movement rules for many natural structural graph parameters.

\section{Preliminaries}

All graphs are finite, undirected, loopless, and simple. We assume throughout that for a graph $G$, $n=|V(G)|\ge1$ and write $m=|E(G)|$. For a positive integer~$q$, write $[q]=\{1,\ldots,q\}$. A $q$-coloring of $G$ is any map $\varphi:V(G)\to[q]$. It need not use every color and need not be proper. It is proper if $\varphi(u)\ne\varphi(v)$ for every $uv\in E(G)$. The color histogram of $\varphi$ is the vector $(|\varphi^{-1}(1)|,\ldots,|\varphi^{-1}(q)|)$.
In running times we will tacitly assume that the integer variable $k$ is at least $2$ so that $\log k\geq 1$.

An instance of \CD is a tuple $(G,q,\varphi,b)$, where $b$ is a nonnegative integer budget. A flip replaces the current color $c$ of any vertex $v$ by an arbitrary color $c' \in [q] \setminus \{c\}$. A swap exchanges the colors of arbitrary distinct vertices $u,v$.
We may disregard swaps between equally colored vertices as they do not change the coloring. A slide is a swap for which $uv\in E(G)$.
The question is whether at most~$b$ permitted moves lead to a proper $q$-coloring.
We write \mbox{\FlipCD}, \mbox{\SwapCD}, and \SlideCD for the three variants. Only the target coloring must be proper.

A vertex cover is a set $C \subseteq V(G)$ such that every edge in $E(G)$ has at least one endpoint in $C$, and $\VC(G)$ is its minimum size. The distance to complete $\DTC(G)$ is the minimum size of a set $S$ such that $G-S$ is a complete graph. The feedback vertex-set number $\FVS(G)$ is the minimum size of a set $F \subseteq V(G)$ such that $G - F$ is a forest. The closure of a rooted forest joins each vertex to all of its ancestors. The treedepth $\TD(G)$ is the minimum, over all rooted forests whose closure contains $G$, of the maximum number of vertices on a root-to-leaf path. The bandwidth is $\BW(G)=\min_\pi\max_{uv\in E(G)}|\pi(u)-\pi(v)|$, where $\pi$ ranges over bijections $V(G)\to[|V(G)|]$, and we define $\BW(G) = 0$ for an edgeless graph. Finally, $\DDP(G)$ is the minimum size of a set whose deletion leaves a disjoint union of paths. For a connected graph $G$, its diameter is the largest distance between two vertices and its radius is the minimum, over all vertices, of the largest distance to another vertex. The domination number is the minimum size of a set of vertices whose closed neighborhoods cover $V(G)$.

We use standard notions from parameterized complexity. A parameterized reduction maps an instance $(\mathcal I,k)$ in time $f(k)|\mathcal I|^{\cO(1)}$ to an equivalent instance $(\mathcal I',k')$ such that $k'\le g(k)$ for some computable functions $f$ and $g$. For integer programming, we use Lenstra's theorem: an integer linear program
(ILP) with~$p$ integer variables and binary encoding length $L$ can be solved
in $f(p)L^{\cO(1)}$ time~\cite{lenstra1983integer}.

\medskip
We will use the following reachability statement between colorings. 
The componentwise reachability criterion for slides is standard in Colored Token Swapping~\cite[discussion preceding Lemma~1]{yamanaka2018colored}. 
The arbitrary-swap criterion is its complete-graph special case. The diameter-dependent bound follows from a classical bound for transposition graphs~\cite[Theorem~1.4.12]{prudden1982transposition}.

\begin{lemma}\label{lem:reachability}
Let $\varphi,\psi:V(G)\to[q]$ be colorings of an $n$-vertex graph.
\begin{enumerate}
\item The coloring $\psi$ is reachable from $\varphi$ by swaps if and only if $|\varphi^{-1}(c)|=|\psi^{-1}(c)|$ for every $c\in[q]$. If it is reachable, $n-1$ swaps suffice.
\item The coloring $\psi$ is reachable from $\varphi$ by slides if and only if these equalities hold separately in every connected component. If they hold, at most \mbox{$\sum_K (|V(K)|-1)$}$(2\diam(K)-1)$ slides suffice, where single-vertex components contribute zero.
\end{enumerate}
\end{lemma}

An equitable $q$-coloring is a proper coloring whose color-class sizes differ by at most one~\cite{meyer1973equitable}. If $n=aq+r$ with $0\le r<q$, exactly $r$ classes have size $a+1$ and $q-r$ classes have size $a$. We use the known $\Wone$-hardness of \EC parameterized by $q+\DDP(G)$~\cite{gomes2023structural}. Gomes et al.\ also prove hardness on interval graphs when $q$, treewidth, and maximum degree are taken together as the parameter~\cite{gomes2019equitable}.
Ordering an interval representation by nondecreasing left endpoint gives bandwidth at most the maximum degree: if $uv$ is an edge with the left endpoint of $u$ first, every interval ordered between $u$ and $v$ intersects $u$. Hence \EC is $\Wone$-hard for $q+\BW(G)$ as well.

\section{Algorithms}

\subsection{Color flipping}

For flips, color names that do not occur initially are interchangeable and we may assume $q\le n$. Further details are provided in~\Cref{app:flip-vc}.

\begin{theorem}\label{thm:flip-vc}
Let $k=\VC(G)$. In
$2^{\cO(k\log k)}(n+m)$ time,
one can compute the minimum number of flips needed to reach a proper
coloring, or determine that no proper $q$-coloring exists.
In particular, \FlipCD is fixed-parameter tractable parameterized
by $\VC(G)$.
\end{theorem}

\begin{proof}
Compute a minimum vertex cover $C$ by the standard bounded-search-tree algorithm, and put $I=V(G)\setminus C$. Enumerate the set partitions $\mathcal P$ of $C$. Keep a partition only if every block is independent and $|\mathcal P|\le q$. 
Such a partition represents which vertices of $C$ receive the same target color.
We aim to find an injection $\pi : \mathcal P \to [q]$ that minimizes the number of color flips necessary to reach a proper coloring $\psi$ with $\psi(v) = \pi(P)$ for all $P \in \mathcal{P}, v \in P$. 

For $x\in I$, let $\mathcal N_{\mathcal P}(x)$ be the set of blocks containing a neighbor of $x$. A coloring represented by $\mathcal P$ extends to $I$ if and only if $|\mathcal N_{\mathcal P}(x)|<q$ for every $x\in I$. Indeed, distinct blocks receive distinct colors, while $I$ is independent.

For a block $P$ and a color $c$, define \mbox{$w(P,c)=|P\setminus\varphi^{-1}(c)|+|\{x\in I: \varphi(x)=c,$} $N(x)\cap P\ne\emptyset\}|$.
Thus, $w(P,c)$ counts the flips forced by assigning color~$c$ to the block~$P$: vertices of~$P$ not already colored~$c$ and color-$c$ vertices of~$I$ adjacent to~$P$ must be recolored. 
If $\pi:\mathcal P\to[q]$ is injective, then the minimum number of flips among extensions using color $\pi(P)$ on $P$ is exactly $\sum_{P\in\mathcal P}w(P,\pi(P))$. The first terms count changes on $C$. A vertex of $I$ must change precisely when its initial color is used by a neighboring block. It is counted exactly once because $\pi$ is injective, and, by the extension condition, it can then be given any available color.

It remains to find a minimum-cost injective assignment of blocks to
colors. Writing $r=|\mathcal P|$, this is a rectangular assignment
problem on the complete bipartite graph between the $r$ blocks and the
$q$ colors. Its weights can be computed in $\cO(m+kq)$ time. Since
$r\le q$, a minimum-cost matching saturating all block vertices can be
found in $\cO(r^2q)\subseteq\cO(k^2q)$ time by the Hungarian method for
unbalanced bipartite graphs~\cite{ramshaw2012weight}. There are at most
$(k+1)^k=2^{\cO(k\log k)}$ partitions of $C$. Taking the best over all
partitions proves the theorem.
\end{proof}

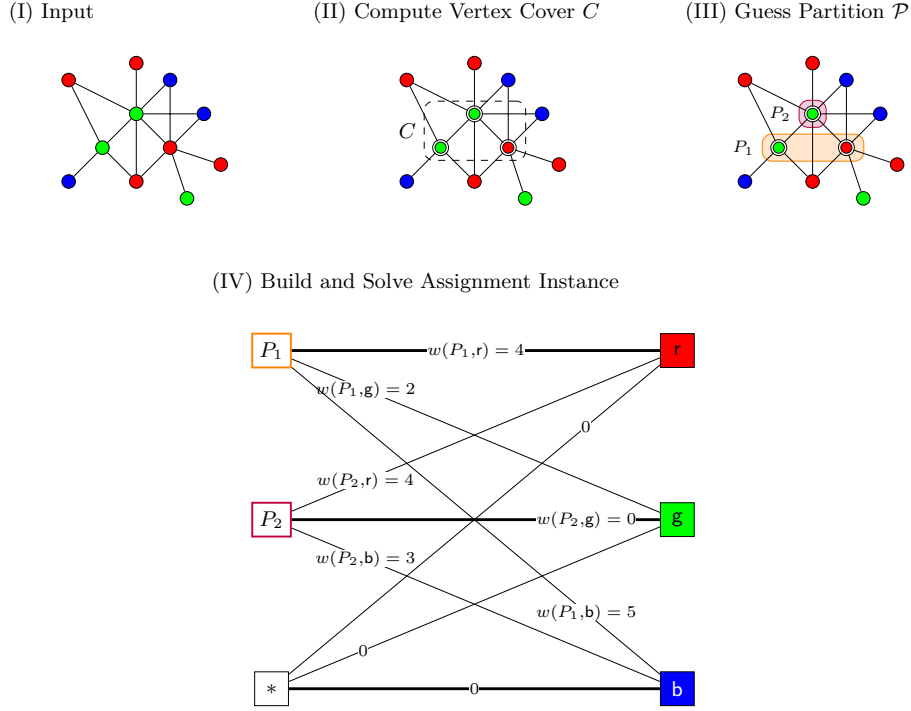
\begin{figure}[t]
\centering
\resizebox{\linewidth}{!}{\input{img/flip-vc-paper}}
\caption{The four stages of the flip algorithm for a fixed vertex cover: the input instance, the computation of $C$, the guessed partition~$\mathcal P$ of $C$ into independent blocks, and the minimum-cost assignment of distinct color names to the blocks. Thick edges show an optimum assignment in the example; the node $\ast$ is a zero-cost dummy used to balance the two sides.}
\label{fig:flip-vc}
\end{figure}

We next consider the parameter distance-to-complete. The algorithm uses the same equality-pattern idea, but the clique left after deleting the modulator requires pairwise distinct colors.

\begin{theorem}\label{thm:flip-dtc}
Let $k=\DTC(G)$. In $2^{\cO(k^2)}n^3$ time, one can compute the minimum number of flips needed to reach a proper coloring, or determine that no proper $q$-coloring exists.
\end{theorem}

\begin{proof}[sketch]
Compute a minimum modulator $S$ and let $C= V(G)\setminus S$. Partition the clique $C$ into the at most $2^k$ classes $\mathcal R$ of equal neighborhoods in $S$. Enumerate the independent-block partitions $\mathcal P$ of $S$. In every proper target coloring, every color occurs at most once on $C$, so a block color is either absent from $C$ or reused by one vertex of a class $R$ anticomplete to that block. We enumerate these choices as maps $\sigma:\mathcal P\to\mathcal R\cup\{\bot\}$ and reject guesses that lack enough vertices in a class of $\mathcal R$ or require more than $q$ colors.

For a surviving guess, create one assignment node for every block and the required number of clique-only nodes of every type. Match these nodes injectively to actual color names. Matching the node for block $P$ to color $q$ indicates that all vertices in $P$ and, if $\sigma(P) = R$ for some $R \in \mathcal R$, a single vertex in $R$ are colored with $q$. Clique-only nodes assign a color to one clique vertex of their respective class. The edge costs count changes on the block and charge zero for a clique vertex exactly when its assigned color already occurs in its type. For each type, distinct already present target colors can be kept on distinct representatives. All remaining assigned colors are assigned bijectively to the remaining vertices. Hence the matching cost is attainable and is also a lower bound on every realization.

There are at most $(k+1)^k$ partitions and $(2^k+1)^k$ maps, and hence $2^{\cO(k^2)}$ cases in total. The assignment side and the color side of each matching instance have at most $n$ nodes each, so weighted matching takes $\cO(n^3)$ time. Computing~$S$ via vertex cover in the complement graph and solving the resulting assignment problems gives the stated running time.
We present the full details in \Cref{app:flip-dtc}. 
\end{proof}

\subsection{Color swapping}

We now fix two colorings $\alpha,\beta:V(G)\to[q]$ and ask for the
minimum number of swaps that transform $\alpha$ into $\beta$.
We assume that
$\alpha$ and $\beta$ have the same color histogram.

The directed color-transition multigraph $D(\alpha,\beta)$ has vertex set
$[q]$. For every vertex $v\in V(G)$,
it contains one arc from $\alpha(v)$ to $\beta(v)$. Parallel arcs are
kept and regarded as distinct. Notice that the vertices of $D(\alpha,\beta)$ represent colors, not
vertices of $G$. Thus, $D(\alpha,\beta)$ has only $q$ vertices but one
arc for every vertex of $G$, and arbitrarily many graph vertices may
give rise to parallel copies of the same arc. In an arc-disjoint cycle decomposition (defined below),
these copies are treated as distinct. We write
$D^{\ne}(\alpha,\beta)$ for the directed multigraph obtained by deleting
all loops.

An \emph{arc-disjoint cycle decomposition} of $D^{\ne}(\alpha,\beta)$ into directed
simple cycles is a collection of directed simple cycles such that every
arc belongs to exactly one cycle. Here a
directed cycle is simple if it visits every color at most once, except
that its first and last colors coincide. Since $\alpha$ and $\beta$ have
the same histogram, every color has equal indegree and outdegree in
$D^{\ne}(\alpha,\beta)$, and hence such an arc-disjoint cycle decomposition exists.

The following statement is the unit-cost interchange-distance
characterization of Amir et al.~\cite[Lemma~1]{amir2009interchange},
stated in our coloring notation.

\begin{lemma}\label{lem:cycle-distance}
Let $M=|E(D^{\ne}(\alpha,\beta))|$, where parallel arcs are counted with
multiplicity, and let $\rho$ be the maximum number of cycles in an arc
decomposition of $D^{\ne}(\alpha,\beta)$ into directed simple cycles.
Then the minimum number of swaps transforming $\alpha$ into
$\beta$ is $M-\rho$.
\end{lemma}

Hence, a maximum-cardinality cycle decomposition determines the optimum swap
distance. Computing such a decomposition is $\NP$-hard when $q$ is unbounded
~\cite[Corollary~1]{amir2009interchange}; Amir et al. also give a
polynomial-time algorithm for every fixed $q$~\cite[Theorem~3]{amir2009interchange}. For two fixed colorings, the
cycle-type ILP below in fact yields fixed-parameter tractability parameterized
by $q$: it uses one variable for each directed simple cycle type on $[q]$.
Note that this does not solve \SwapCD parameterized by $q$ alone, since its proper target
coloring is not given. Indeed, \SwapCD remains $\NP$-complete for every fixed
$q\ge3$~\cite{fellows2026solution}. The additional vertex-cover parameter lets us
encode all possible proper targets by boundedly many neighborhood types, while
the cycle part of the ILP depends only on $q$.

\begin{theorem}\label{thm:swap-vc}
\SwapCD is fixed-parameter tractable parameterized by $\VC(G)+q$.
\end{theorem}

\begin{proof}
Compute a minimum vertex cover $C$ of size $k$ by the standard bounded-search-tree algorithm, and put $I=V(G)\setminus C$. Enumerate the $q^k$ maps $\gamma:C\to[q]$ and discard those that are not proper on $G[C]$. Partition $I$ into the at most $2^k$ classes~$R$ of equal neighborhoods in $C$.

We write the following ILP (the full details are presented in \Cref{app:ilp}).
For a fixed $\gamma$, introduce a nonnegative integer variable $x_{R,a,c}$ for every class~$R$ and colors $a,c$. It counts vertices in $R$ that initially have color $a$ and finally have color~$c$. We require $\sum_c x_{R,a,c}=|R\cap\varphi^{-1}(a)|$. We set $x_{R,a,c}=0$ whenever $c$ occurs on~$N(R)$ under~$\gamma$.

Together with the fixed transitions $\varphi(v)\to\gamma(v)$ on $C$, the variables determine a multiplicity
$T_{a,c}=|\{v\in C:\varphi(v)=a,\gamma(v)=c\}|+\sum_Rx_{R,a,c}$
for every ordered color pair.
These variables model the arcs of transition multigraph.
Let~$\Gamma_q$ be the set of directed simple cycles on $[q]$, including loops, and introduce a nonnegative integer variable $y_Z$ for every $Z\in\Gamma_q$. For every ordered pair $(a,c)$ require
\mbox{$\sum_{Z:(a,c)\in E(Z)}y_Z=T_{a,c}$}, and impose the budget constraint
\mbox{$\sum_{Z\in\Gamma_q}(|Z|-1)y_Z\le b$}.

The supply and forbidden-color constraints describe exactly the proper extensions of $\gamma$ to the independent set $I$. The cycle equations ensure that the transition multigraph is the union of directed simple cycles and therefore enforce equality of the initial and target color multiplicities. Every feasible choice of the $y$-variables specifies a swap sequence of length $\sum_{Z\in\Gamma_q}(|Z|-1)y_Z$, so the budget constraint implies that the target is reachable within $b$ swaps. Conversely, the transitions of every proper target reachable within at most~$b$ swaps satisfy the $x$-constraints.
Let $\psi$ be the coloring prescribed by $\gamma$ and the $x$-variables.
Choosing a decomposition of $D^{\ne}(\varphi, \psi)$ with the maximum number of directed simple cycles and taking every loop as a one-vertex cycle supplies $y$-variables whose budget expression is, by \Cref{lem:cycle-distance}, the exact swap distance. Thus the ILP is feasible exactly for the desired targets. 

There are $|\Gamma_q|=q+\sum_{\ell=2}^q\binom q\ell(\ell-1)!<e\,q!$ variables of type
$y$ and at most $2^kq^2$ variables of type $x$. By \Cref{lem:reachability}, we may replace $b$ by
$\min\{b,n-1\}$. Every coefficient and right-hand side then has
$\cO(\log n+\log q)$ bits, and the total ILP encoding length is
$f(k+q)n^{\cO(1)}$ for a computable function $f$. Lenstra's theorem, applied
to each of the $q^k$ choices for $\gamma$, gives a running time
$f'(k+q)n^{\cO(1)}$.
\end{proof}

We remark that this transition-multigraph approach does not seem to extend to slides, since the available moves depend on the positions of the colors in the graph. We therefore could not resolve the complexity of \SlideCD with respect to $\VC(G)+q$ and leave it as an open question.

\section{Hardness Results}

We now turn to further natural structural graph parameters and show hardness with respect to these. 
The full details of the longer reductions are given in \Cref{app:hardness}.

\subsection{Flipping: treedepth and feedback vertex set}

Fiala, Golovach, and Kratochv{\'i}l proved that \LC is
$\Wone$-hard parameterized by vertex cover~\cite{fiala2011coloring}.
We reduce from their problem using a list-enforcement gadget that
simultaneously bounds the treedepth and the feedback vertex set of the
constructed graph in terms of the vertex-cover number of the source
instance.

\begin{theorem}\label{thm:flip-hard}
\FlipCD is $\Wone$-hard parameterized by $\TD(G)+\FVS(G)$.
\end{theorem}

\begin{proof}[sketch]
Let $(G,q,L)$ be a \LC instance. First, delete colors that occur in no list. Let $n=|V(G)|$ and set $b=n$. Construct a graph $H$ from $G$ as follows. Initially color every original vertex with color $1$. For every $v\in V(G)$ and every forbidden color $c\notin L(v)$, attach $b+1$ pendent vertices to $v$ and color them~$c$. A list coloring is reached by changing only the original vertices. Conversely, if a reachable proper coloring gives $v$ a forbidden color $c$, all its $b+1$ $c$-colored pendent vertices must have changed, which exceeds the budget.

If $C$ is a vertex cover of the source graph, deleting $C$ from the constructed graph leaves a star forest, so $\FVS(H)\le |C|$. For treedepth, put the vertices of $C$ on a chain, put every original vertex outside $C$ below the chain, and put each pendent vertex below its center. This gives $\TD(H)\le |C|+2$. See \Cref{app:list} for the complete argument.
\end{proof}

\subsection{Swapping and sliding: bandwidth and distance to paths}

We use reductions from \EC for both the swap and slide variants. For the
swap variant, the input graph itself suffices. For the slide variant,
connectivity must be ensured without destroying the parameter bound.

\begin{theorem}\label{thm:structural-hard}
Each of \SwapCD and \SlideCD is $\Wone$-hard parameterized by
$q+\BW(G)$ and by $q+\DDP(G)$. For the parameter $q+\DDP(G)$, hardness
holds even on graphs with a dominating vertex, and hence on graphs of
diameter at most two.
\end{theorem}

\begin{proof}[sketch]
Given an \EC instance $(G,q)$ with $n=aq+r$, initially color the vertices
of $G$ so that $r$ colors have multiplicity $a+1$ and all other colors
have multiplicity $a$. Every proper target reachable by swaps is an
equitable coloring of $G$, since swaps preserve these multiplicities.
Conversely, every equitable coloring can be relabeled to have the
prescribed multiplicity for each color and, by
\Cref{lem:reachability}, can then be reached in at most $n-1$ swaps.
This proves the claimed $\Wone$-hardness of \SwapCD for both
parameterizations, without changing the graph.

For the distance-to-disjoint-paths lower bound for both the swap and
slide variants, add a universal vertex $z$ of a new, unique color. Arbitrary transpositions still work for swaps and are simulated for slides through the path $u,z,v$. In a proper target, the color of $z$ is unique.  Exchanging two color names if necessary yields an equitable coloring of $G$. Removing $z$ recovers $G$, so the parameter increases by at most two.

For the slide variant parameterized by $q+\BW(G)$, attach a chain of $n$ rainbow $(q+1)$-cliques and connect every component of $G$ to a suitable clique. The chain makes the graph connected, while its forced rainbow coloring reserves exactly $n$ copies of each color. A folded bandwidth order places each component connector last and stretches every original distance by a factor of at most two, giving bandwidth at most $2(q+2)(\BW(G)+1)$. In \Cref{app:equitable} we give the full construction and verify all histogram, reachability, and bandwidth claims.
\end{proof}

\begin{figure}[t]
\centering
\resizebox{0.88\linewidth}{!}{\input{img/slide-bw}}
\caption{The connectivity gadget for the sliding reduction, shown for $q=3$. Rainbow $(q+1)$-cliques form a chain, and a connector of each component of $G$ is attached to a distinct clique in the chain.}
\label{fig:slide-bandwidth}
\end{figure}
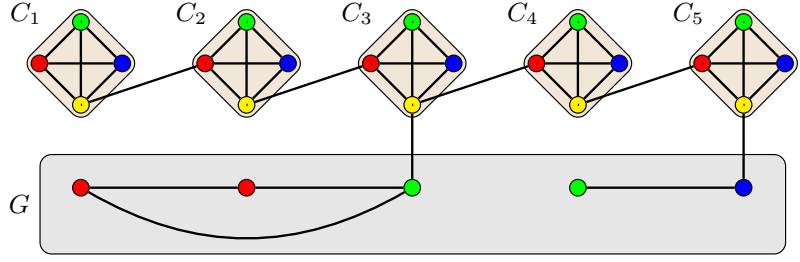

\subsection{A four-color and a uniform-base lower bound}

The preceding reductions bound structural parameters by computable functions of the input parameter.
By a reduction from the classical $q$-\textsc{Coloring} problem, we can prove $\NP$-hardness even for a fixed value of the parameter diameter and obtain lower bounds under the ETH on the runtime parameterized by the vertex cover number.

\begin{lemma}\label{lem:coloring-reduction}
    Let $q \ge 3$. The $q$-\textsc{Coloring} problem can be reduced in polynomial time to \CD in any model.
\end{lemma}

\begin{proof}
Let $(G,q)$ be a $q$-Coloring instance on $n$ vertices. Assume $q<n$, since the remaining instances are trivial. Let $G$ have $n\ge2$ vertices. Add $q$ independent reservoirs $R_1,R_2,\dots,R_q$, each of size $n$, and a vertex $z$ adjacent to every vertex of $V(G)\cup R_1\cup R_2\cup \dots \cup R_q$. There are no other edges incident with reservoir vertices. Initially, $z$ has color $q+1$, every vertex of $G$ has color $1$, and every vertex of~$R_i$ has color $i$.

For flips and swaps, set the budget to $n$. Given a proper coloring of $G$, process its vertices one by one. Each incorrectly colored vertex is recolored to its desired color, which can be done by flipping its color for \FlipCD and by swapping with an unused reserve vertex of its desired color for \SwapCD. In the reverse direction, every original color has multiplicity at least~$n\ge2$, whereas color~$q+1$ is unique. Properness therefore forces $z$ to retain color~$q+1$, and the restriction to $G$ is a $q$-coloring.

For slides, use budget $3n$. Each reserve swap is simulated through $z$ by three slides. 
The reverse direction follows by the same argument as for flips and swaps. The constructed size is polynomial because $q<n$.
\end{proof}

Membership in $\NP$ of \CD has been established in all three models~\cite{debiasi2019restoring,fellows2026solution,garnero2018fixing}. The constructed graph is connected and has a dominating vertex. Two reservoir vertices are nonadjacent, so its diameter is exactly two, and $z$ is adjacent to all vertices, providing for a radius of 1. As \textsc{Coloring} is $\NP$-complete even for $q = 3$~\cite{garey1979computers}, we obtain the following theorem:

\begin{theorem}\label{thm:diameter}
    For each of flipping, swapping, and sliding, \CD is $\NP$-complete on connected graphs of diameter exactly two, even with exactly four available colors. Consequently, on connected input graphs, each variant is $\paraNP$-complete parameterized by $q+\diam(G)$, and also by $q$ together with either radius or domination number.
\end{theorem}

Jaffke and Jansen proved that, unless the Exponential Time Hypothesis (ETH) fails~\cite{impagliazzo2001eth}, there is no constant $\theta$ that gives an $\cO(\theta^k n^{\cO(1)})$ algorithm for $q$-Coloring parameterized by $k=\VC(G)$ simultaneously for all fixed $q$~\cite{jaffke2023finegrained}.
Additionally, Lokshtanov et al. have shown that, assuming the ETH holds, there is no $2^{o(k \log k)} n^{\cO(1)}$ time algorithm using the same parameterization~\cite{lokshtanov18superexp}.
The construction from \Cref{lem:coloring-reduction} increases the vertex cover number of the graph by one, as every newly added edge is incident with $z$. Thus, the lower bounds extend to all three discovery models.

\begin{theorem}\label{thm:eth}
For each variant $X\in\{\FlipCD,\SwapCD,\SlideCD\}$, there is no universal constant $\theta$ such that, for every fixed $q$, the restriction of $X$ to $q$ colors can be solved in $\cO(\theta^{k}n^{\cO(1)})$ time where $k = \VC(G)$, unless the ETH fails.
Moreover, there is no $2^{o(k \log k)} n^{\cO(1)}$ algorithm assuming the ETH holds.
\end{theorem}

\section{Conclusion}

The three modification step models lead to different complexity results. Under flips, the vertices outside a small modulator can be treated independently for calculating the number of required color flips.
Under swaps, the size of color classes is preserved and a cycle decomposition of the color-transition multigraph permits calculating the exact swap distance through our ILP.
Under slides, transition counts no longer determine the relevant distance, since the positions of the colors inside the graph matter as well.
In the other direction, the common diameter-two reduction shows that even radius one and the presence of a dominating vertex do not by themselves make the repair task easy.

Several natural questions remain open. The first is whether \SlideCD is
fixed-parameter tractable for $\VC(G)+q$. We also leave open if the number of colors be removed from the parameter of the swap algorithm. Finally, it is worthwhile to establish whether the flip algorithm for distance to complete can be extended to the distance to cluster (that is, the disjoint union of cliques), where
several cliques may each use the same color once.

\label{main:last}
\clearpage
\bibliographystyle{splncs04}
\bibliography{references}

\clearpage
\appendix
\renewcommand{\theHsection}{A\arabic{section}}
\crefalias{section}{appendix}
\crefalias{subsection}{appendix}
\section{Detailed Algorithmic Proofs}\label{app:algorithms}

\subsection{The vertex-cover assignment algorithm}\label{app:flip-vc}

We give the full correctness and running-time argument for
\Cref{thm:flip-vc}.
We can assume that $q \leq n$.
Let~$\varphi'$ be a coloring realizing a minimum number of flips and suppose~$\varphi$ and~$\varphi'$ use $n < p \leq q$ distinct colors in total.
If~$s$ is the number of colors used by both~$\varphi$ and~$\varphi'$, and~$\varphi'$ uses~$r$ colors unused by~$\varphi$, then $s + r \leq n$ as there are only~$n$ vertices to color.
We can then injectively map the~$r$ $\varphi'$-only colors first to the $\varphi$-only colors and then to colors not used by $\varphi$ if necessary and thereby ensure the total number of colors used is at most~$n$.
This does not increase the number of color flips as all vertices colored with a $\varphi'$-only color needed to be recolored anyway. 
We can therefore assume that the number of available colors is at most~$n$.
Let $C$ be a minimum vertex cover of size $k$ and
let $I=V(G)\setminus C$. A target coloring induces a partition
$\mathcal P$ of $C$ into its nonempty color classes. Each block is
independent, different blocks receive different color names, and
$|\mathcal P|\le q$. Conversely, every such partition together with an
injective map $\pi:\mathcal P\to[q]$ specifies a proper coloring on $C$.

Fix an independent-block partition $\mathcal P$. For $x\in I$, let
$\mathcal N_{\mathcal P}(x)$ be the set of blocks $P$ with $N(x) \cap P \neq \emptyset$. Since
$I$ is independent, an injection $\pi$ extends to a proper coloring of $G$
if and only if $|\mathcal N_{\mathcal P}(x)|<q$ for every $x\in I$. The
condition is independent of the actual color names: the neighbors of $x$
use exactly $|\mathcal N_{\mathcal P}(x)|$ distinct colors. If this number
is smaller than $q$, at least one color remains available for~$x$.

\begin{lemma}\label{lem:fixed-partition-cost}
For a fixed admissible partition $\mathcal P$ and injection
$\pi:\mathcal P\to[q]$, the minimum number of flips among all extensions is
$\sum_{P\in\mathcal P}w(P,\pi(P))$, where
$w(P,c)=|P\setminus\varphi^{-1}(c)|
+|\{x\in I:\varphi(x)=c,\ N(x)\cap P\ne\emptyset\}|$.
\end{lemma}

\begin{proof}
The first term counts exactly the vertices of $C$ whose color changes.
Consider $x\in I$ and put $a=\varphi(x)$. The vertex $x$ can keep its color
if and only if no neighboring block is assigned $a$. Since $\pi$ is
injective, at most one block $P$ satisfies $\pi(P)=a$. Hence, if $x$ has to
change, it is counted exactly once by the second term of
$w(P,\pi(P))$. If it does not have to change, it is not counted.

It remains to see that all counted changes can be made simultaneously. For
every uncounted vertex of $I$, retain its initial color. For every counted
vertex~$x$, choose any color not used by the blocks in
$\mathcal N_{\mathcal P}(x)$. Such a color exists by admissibility.
Choices for different vertices of $I$ cannot conflict because $I$ is
independent. Thus the lower bound given by the sum is attained.
\end{proof}

The optimal injection $\pi$ is a rectangular assignment problem with left
side~$\mathcal P$, right side $[q]$, and edge costs $w(P,c)$. If desired,
one may add $q-|\mathcal P|$ dummy block vertices with zero-cost incident
edges and solve a square minimum-weight perfect matching instance. For one
partition, all neighborhood sets and all second cost terms are computed by
scanning the edges once while ensuring that no vertex $x \in I$ is counted more than once for any fixed assignment of color $c$ to an adjacent block $P$. The first cost terms take $\cO(kq)$ time, and the
assignment takes $\cO(k^2q)$ time. There are at most \mbox{$(k+1)^k=2^{\cO(k\log k)}$} partitions of $C$.
Testing independence and the extension condition as well as building and solving the assignment problem takes $\cO(m+k^2q)$ time per partition.

A minimum vertex cover can be found by an iterative branching algorithm. Up to depth $k$, this costs
$\cO(2^k(n+m))$ time. We have $q\leq n$. We therefore obtain a total running time of
\[
\cO\bigl(2^k(n+m)+(k+1)^k(m+k^2q)\bigr)
=
2^{\cO(k\log k)}(n+m).
\]

Since the palette preprocessing reduces $q$ to at most $n$, the algorithm
is fixed-parameter tractable in $k$ and either computes the exact optimum
or reports that no proper $q$-coloring exists.

\subsection{The distance-to-complete algorithm}\label{app:flip-dtc}

We now give further details for \Cref{thm:flip-dtc}. Let $S \subseteq V(G)$ be a
minimum set such that $C= V(G) \setminus S$ is a clique. It can be found as a
minimum vertex cover of the complement graph in
$\cO(2^kn^2)$ time where $k = |S|$ . Two vertices of $C$ have the same \emph{type} if they
have the same neighborhood in $S$. Let $\mathcal R$ denote the resulting
partition of~$C$. 
Clearly $|\mathcal R|\le2^k$.

Fix an independent-block partition $\mathcal P$ of $S$. In any proper
coloring, distinct blocks of $\mathcal P$ receive distinct colors and all
vertices of the clique $C$ receive distinct colors. A color used by a block
$P$ is either absent from $C$, or it appears on exactly one clique vertex.
In the latter case that vertex must belong to a type~$R$ anticomplete to
$P$. We therefore guess a map
$\sigma:\mathcal P\to\mathcal R\cup\{\bot\}$. The value $\bot$ says that
the color of $P$ is absent from $C$. 
The value~$R$ says that one vertex of~$R$ reuses this color.

Put $a_R=|\sigma^{-1}(R)|$ and
$s=|\{P\in\mathcal P:\sigma(P)\ne\bot\}|$. A guess is feasible only if
every $P$ with $\sigma(P)=R$ is anticomplete to $R$, if $a_R\le|R|$ for
every type, and if $|C|+|\mathcal P|-s\le q$. The last expression is the
number of distinct colors required: without reuse, the clique and the
blocks use disjoint color sets. Each non-$\bot$ value saves exactly one
color.

\begin{figure}[!ht]
\centering
\resizebox{\linewidth}{!}{\input{img/flip-dtc}}
\caption{A guess $\sigma$ for the distance-to-complete algorithm. For each block $P$ of the partition of $S$, the equality $\sigma(P)=R$ means that $R$ is anticomplete to $P$ and that one vertex of $R$ reuses the color assigned to $P$; the equality $\sigma(P)=\bot$ means that this color is absent from $C$. Solid arrows show one guess and dashed arrows other admissible choices.}
\label{fig:flip-dtc}
\end{figure}
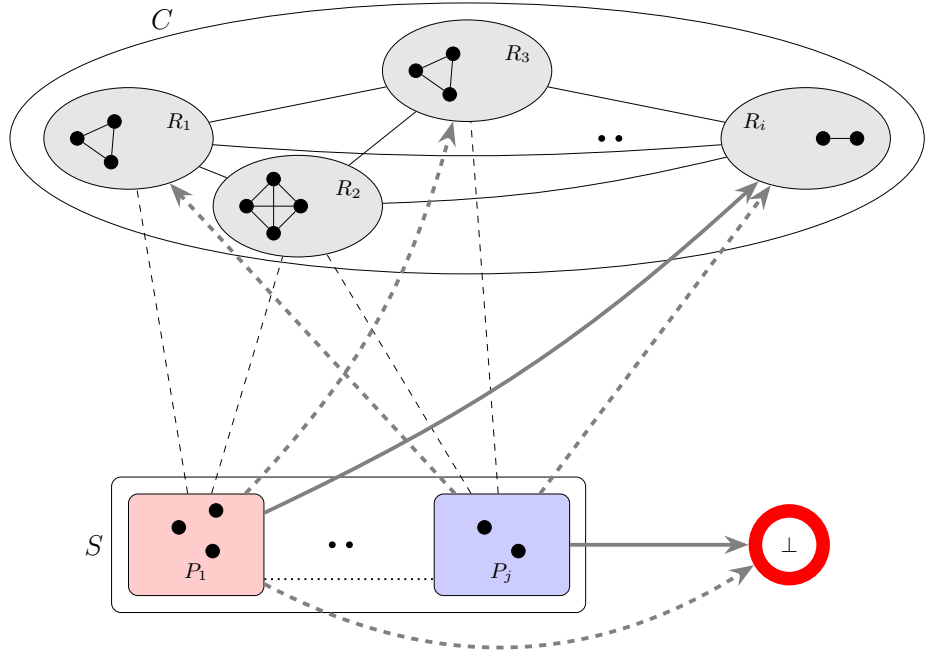

For a feasible guess, create one assignment node for each block
$P\in\mathcal P$. If $\sigma(P)=R$, this node represents both $P$ and the
one clique vertex of type $R$ that reuses its color. In addition, create
$|R|-a_R$ clique-only nodes for each type $R$. The total number of nodes is
$|C|+|\mathcal P|-s$, so the feasibility test ensures that they can be
matched injectively to colors in $[q]$.

The cost of assigning color $c$ to a block node $P$ is
$|P\setminus\varphi^{-1}(c)|$. If $\sigma(P)=R$, add zero when $R$
contains an initially $c$-colored vertex and add one otherwise. The cost of
assigning $c$ to a clique-only node of type $R$ is zero when $c$ already
occurs in $R$ and one otherwise. A minimum-cost injection can again be
found by weighted bipartite matching.

\begin{lemma}\label{lem:dtc-realizable}
For a fixed feasible pair $(\mathcal P,\sigma)$ and a fixed injective
assignment of color names, the sum of the above costs is exactly the minimum
number of flips realizing that assignment.
\end{lemma}

\begin{proof}
The block terms plainly count the forced changes in $S$. Fix a type $R$ and
let~$T_R$ be the set of target colors assigned to its vertices, including
colors represented by block nodes mapped to $R$. For every $c\in T_R$ that
already occurs in~$R$, retain one initially $c$-colored vertex. These chosen
vertices are distinct because the colors in $T_R$ are distinct. Match the
remaining target colors bijectively to the remaining vertices of $R$. Each
such vertex changes color. Hence exactly one vertex can be kept for every
target color already present in $R$, and all other vertices change. This is
precisely the cost charged by the nodes of type $R$.

Conversely, in any realization, at most one vertex of $R$ can keep color
$c$ because the target is proper on the clique. Such a vertex can exist only
when $c$ occurs initially in $R$. Thus, no realization has smaller cost.
The construction is independent for the different types and therefore
realizes all zero-cost representatives simultaneously.

The resulting target coloring is proper. Every block of $\mathcal P$ is independent, distinct assignment nodes receive distinct colors, and all vertices of $C$ correspond to distinct assignment nodes. The colors shared between $S$ and $C$ are precisely those assigned to blocks $P$ with $\sigma(P)\ne\bot$. Each such color is shared only by~$P$ and the clique vertex represented by the same block node; that vertex lies in~$\sigma(P)$, which is anticomplete to $P$ by feasibility.
\end{proof}

There are at most $(k+1)^k \in k^{\cO(k)}$ partitions $\mathcal P$,
and for each partition there are at most
$(2^k+1)^{|\mathcal P|}\le (2^k+1)^k$ maps~$\sigma$.
For each feasible guess, the assignment instance has at most~$n$
vertices and, after the palette preprocessing, at most $n$ color names. It can therefore be solved in $\cO(n^3)$ time. Together with the computation of $S$, and using
$n^2\le n^3$, $k^{\cO(k)}=2^{\cO(k\log k)}$, and
$2^k+1\le 2^{k+1}$, the running time is
\begin{align*}
\cO\bigl(2^kn^2+k^{\cO(k)}(2^k+1)^kn^3\bigr)
&\le
\cO\bigl((2^k+2^{\cO(k\log k)}\cdot 2^{k(k+1)})n^3\bigr)\\
&=
\cO\bigl((2^k+2^{\cO(k^2+k\log k+k)})n^3\bigr)\\
&=
2^{\cO(k^2)}n^3\bigr.
\end{align*}

\section{The Color-Cycle ILP in Detail}\label{app:ilp}

We now prove the equivalence used in \Cref{thm:swap-vc}. Fix \mbox{$\gamma:C\to[q]$} such that $\gamma$ is a proper coloring. For a neighborhood class $R\subseteq I$, let $A_{R,a}=|R\cap\varphi^{-1}(a)|$. The ILP contains the supply equations $\sum_c x_{R,a,c}=A_{R,a}$ and the nonnegativity constraints \mbox{$x_{R,a,c}\in\Znonneg$}. If $c\in\gamma(N(R))$, it also contains $x_{R,a,c}=0$.

Define $T_{a,c}=|\{v\in C:\varphi(v)=a,\gamma(v)=c\}|+\sum_Rx_{R,a,c}$. For every directed simple color cycle $Z$, including each one-vertex loop, let $y_Z\in\Znonneg$. The arc equations are $\sum_{Z:(a,c)\in E(Z)}y_Z=T_{a,c}$ for all $a,c\in[q]$. Since a union of directed cycles is balanced at every color, these equations imply $\sum_cT_{a,c}=\sum_cT_{c,a}$. The left side is the initial multiplicity of $a$ and the right side is its target multiplicity. Thus swap reachability is enforced without a separate histogram constraint.

Given a feasible integer solution, assign final colors within each class $R$ according to the counts $x_{R,a,c}$. This is possible because the supply equations account for every vertex in each initial-color subclass. The forbidden-color equations make the resulting coloring proper across edges from $C$ to $I$. Properness on $C$ was checked when $\gamma$ was guessed, and $I$ has no edges. The $y$-variables decompose its transition multigraph, and the budget inequality plus \Cref{lem:cycle-distance} supplies a swap sequence of length at most $b$.

In the other direction, let a proper target reachable within at most~$b$ color swaps be given. Counting its transitions within every class defines the $x$-variables. Properness gives all forbidden-color equations. Equal histograms make the transition multigraph balanced, so its arcs decompose into directed simple cycles. Choose a decomposition that realizes the minimum swap distance. \Cref{lem:cycle-distance} shows that its $y$-variables satisfy the budget constraint. This proves both implications.

\enlargethispage{3\baselineskip}
There are at most $2^kq^2$ variables of type $x$ and $q+\sum_{\ell=2}^q\binom q\ell(\ell-1)!<e\,q!$ variables of type $y$. By \Cref{lem:reachability}, replace $b$ by $\min\{b,n-1\}$. Hence the right-hand sides are polynomial and the encoding length is bounded by $\cO\big(f(k+q) \cdot n^{\cO(1)}\big)$ for a computable function $f$. Across the $q^k$ choices for $\gamma$, Lenstra's algorithm proves fixed-parameter tractability.

\section{Complete Hardness Proofs}\label{app:hardness}

\subsection{List enforcement for flipping}\label{app:list}

We prove \Cref{thm:flip-hard} in full detail. 
Let $(G,q,L)$ be an instance of \LC, and let $C$ be a vertex cover of $G$ of size $k$. The empty graph is a trivial yes-instance and can be handled separately. We delete colors that occur in no list and relabel the remaining palette. If some vertex has an empty list, we return a fixed no-instance. Thus, $q$ is at most the explicit list-encoding length. Put $n=|V(G)|$ and $b=n$. Construct $H$ from a copy of $G$. Give every original vertex initial color~$1$. For every original vertex $v$ and every color $c\in[q]\setminus L(v)$, add $b+1$ pendent vertices adjacent only to $v$, all initially colored $c$. This is a polynomial construction with the same palette.

\begin{figure}[!ht]
\centering
\resizebox{0.95\linewidth}{!}{\input{img/flip-hard}}
\caption{The list-enforcement gadget, illustrated for $b=4$. For every
$c \in [q] \setminus L(v)$, the vertex $v$ receives $b+1$ pendent vertices initially
colored $c$. Dotted segments abbreviate omitted vertices.}
\label{fig:list-enforcement}
\end{figure}
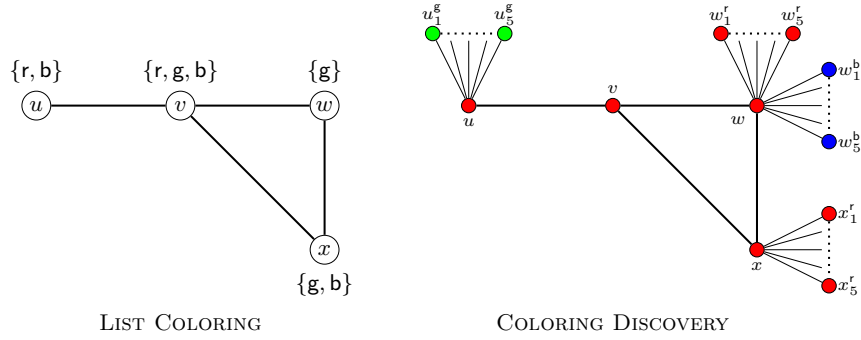
\FloatBarrier

If $\psi$ is a proper list coloring of $G$, flip each original vertex whose initial color differs from $\psi(v)$ and leave all pendent vertices unchanged. At most $n=b$ flips are used. Edges inside the copy of $G$ are proper because $\psi$ is proper. If a pendent vertex at $v$ has color $c$, then $c\notin L(v)$ and hence $\psi(v)\ne c$, so its edge is proper as well.

Conversely, suppose a proper coloring of $H$ is reachable with at most $b$ flips. If an original vertex $v$ ends with a forbidden color $c$, each of its $b+1$ pendent vertices initially colored $c$ must end with another color. This already requires more than $b$ flips. Thus every original vertex receives a color from its list, and restriction to~$G$ is a proper list coloring.

It remains to bound the parameters. The graph $G-C$ is independent. After deleting $C$ from $H$, every original vertex outside $C$ is the center of a star of pendent vertices, while pendent vertices whose center belongs to $C$ become isolated. Hence $H-C$ is a forest and $\FVS(H)\le k$.

For treedepth, arrange the vertices of $C$ on one root-to-leaf chain. If $k>0$, make every original vertex outside $C$ a child of the last vertex of this chain. Every edge of $G$ has an endpoint in $C$, so its endpoints are comparable in the resulting rooted tree. Attach a pendent vertex of a vertex outside $C$ below that vertex. Attach a pendent vertex of a vertex in $C$ directly below its center. If $k=0$, instead take every original vertex as a root and hang its pendent vertices below it. All edges of $H$ then join comparable vertices, and the height is at most $k+2$. Thus $\TD(H)+\FVS(H)\le2k+2$, completing the parameterized reduction from the $\Wone$-hard \LC problem of Fiala et al.~\cite{fiala2011coloring}.

\subsection{Reductions from Equitable Coloring}\label{app:equitable}

Let $(G,q)$ be an \EC instance with $n$ vertices. We may assume $q\le n$, since instances with $q>n$ are trivially positive and can be mapped to a fixed yes-instance. Write $n=aq+r$, where $0\le r<q$. Choose any initial coloring~$\varphi$ in which colors $1,\ldots,r$ occur $a+1$ times and the remaining colors occur $a$ times.

\paragraph{Arbitrary swaps.}
Use the instance $(G,q,\varphi,n-1)$. If $G$ has an equitable coloring, rename its colors so that the $r$ large color classes receive names $1,\ldots,r$. It now has exactly the initial histogram and is reachable by \Cref{lem:reachability}. Conversely, swaps preserve every color-class size, so any reachable proper coloring is equitable. The graph, and hence both bandwidth and distance to disjoint paths, is unchanged. The known hardness of \EC proves the two claims for arbitrary swaps.

\paragraph{A cone for distance to disjoint paths.}
Create $H$ by adding a universal vertex~$z$ to $G$, use the palette $[q+1]$, and set $\varphi(z)=q+1$. For swaps use budget~$n$. 
For slides use budget $3n$. If $\psi$ is an equitable coloring of $G$, rename it to the initial histogram and keep color $q+1$ on $z$. Choose a token matching that fixes~$z$. The induced permutation on $V(G)$ admits a factorization into at most \mbox{$n-1$} transpositions, so the swap budget suffices. For slides, perform every such transposition of \mbox{$u,v\in V(G)$} by $(u,z),(z,v),(u,z)$. This restores the color of $z$, and $3n$ moves suffice.

For the reverse implication, let $\psi$ be a reachable proper coloring of $H$. Its color histogram equals the initial one. Since $z$ is universal, $\psi(z)$ occurs exactly once. If $\psi(z)=q+1$, restriction to $G$ is immediately an equitable $q$-coloring. Otherwise, let $\psi(z)=i\in[q]$. Then the initial multiplicity of $i$ was one, and the unique token of color $q+1$ lies in $G$. Exchanging the names $i$ and $q+1$ in $\psi$ gives a proper coloring with $z$ colored $q+1$ whose restriction to $G$ has precisely the original equitable multiplicities. Finally, if $S$ is a modulator of $G$ to disjoint paths, then $S\cup\{z\}$ is such a modulator of $H$. Hence $\DDP(H)\le\DDP(G)+1$, the new palette has size $q+1$, and $z$ is dominating.

\paragraph{Sliding with bounded bandwidth.}
For slides, we make the construction connected as follows. For each $j\in[n]$, add a clique $C_j$ of size $q+1$ and initially use every color in $[q+1]$ once on it. For $j<n$, join the color-$(q+1)$ vertex of $C_j$ to the color-$1$ vertex of $C_{j+1}$. Let $G_1,\ldots,G_s$ be the connected components of~$G$ in an arbitrary fixed order, choose a connector $v_i\in V(G_i)$, and put $r(i)=\sum_{h\le i}|V(G_h)|$. Join~$v_i$ to the color-$(q+1)$ vertex of $C_{r(i)}$. The initial coloring on $G$ is $\varphi$ constructed as above, and no vertex of $G$ has color $q+1$. Let $N=|V(H)|=n(q+2)$ and set the budget to~$2N^2$.

If $G$ has an equitable coloring $\psi$, rename its colors to match the histogram of~$\varphi$ and leave every clique unchanged. The resulting coloring of $H$ is proper: each clique is rainbow, consecutive clique endpoints have colors $q+1$ and $1$, and every connector in $G$ uses a color in $[q]$. It has the same global histogram as the initial coloring. The graph $H$ is connected, so by \Cref{lem:reachability} it can be reached with at most $(N-1)(2\diam(H)-1)<2N^2$ slides.

Conversely, slides preserve the global histogram, and because $H$ is connected every target with equal histogram is slide-reachable by~\Cref{lem:reachability}. In every proper target, each $(q+1)$-clique uses all $q+1$ colors once. The $n$ cliques therefore consume exactly $n$ copies of every color. What remains on $G$ is exactly the initial equitable histogram on $[q]$ and no copy of color $q+1$. Restriction to~$G$ is consequently an equitable $q$-coloring.

We finish by bounding bandwidth. Let $B=\BW(G)$. For each component, take the relative order induced by a bandwidth-$B$ layout and concatenate the component orders in the fixed order $G_1,\ldots,G_s$. This does not increase the length of any edge. Fold the order of $G_i$ so that its connector $v_i$ becomes last. 
Write the old component order as $x_1,\ldots,x_a,v_i,y_1,\ldots,y_d$. If $a\ge d$, first output $x_1,\ldots,x_{a-d}$, then, for $h=1,\ldots,d$, output $x_{a-d+h},y_{d-h+1}$, and finally output~$v_i$. 
If $d>a$, first output $y_d,\ldots,y_{a+1}$, then, for $h=1,\ldots,a$, output $y_{a-h+1},x_h$, and finally output $v_i$. This also covers $a=0$ or $d=0$. Every two consecutive vertices in the old order that lie within the same connected component are now at distance at most two, so every pairwise order distance, and in particular every edge length, grows by at most a factor of two.

Number the resulting order of $G$ by $1,\ldots,n$. Put its $j$th vertex at position $(q+2)j$. Put the $q+1$ vertices of $C_j$ in positions $(q+2)(j-1)+1,\ldots,(q+2)j-1$, with its color-$(q+1)$ vertex last and its color-$1$ vertex first. Clique edges have length at most $q$, chain edges have length two, and a component connector is adjacent to the clique vertex immediately before it. An original edge has length at most $2(q+2)B$. Therefore $\BW(H)\le2(q+2)(B+1)$, which completes the parameterized reduction.

\end{document}

%% file: img/flip-vc-paper.tex
\begin{tikzpicture}

% Input Graph
\node[text node] (text1) at (0,0) {(I) Input};
\node[small vertex, fill=green] (a) at ($(text1.west) + (1.5,-2)$) {};
\node[small vertex, fill=green] (b) at ($(a) + (0.5,0.5)$) {};
\node[small vertex, fill=red] (c) at ($(a) + (1,0)$) {};
\node[small vertex, fill=blue] (a1) at ($(a) + (-0.5,-0.5)$) {};
\node[small vertex, fill=red] (ab1) at ($(a) + (-0.5,1)$) {};
\node[small vertex, fill=red] (b1) at ($(b) + (0,0.75)$) {};
\node[small vertex, fill=blue] (bc1) at ($(b) + (0.5,0.5)$) {};
\node[small vertex, fill=blue] (bc2) at ($(c) + (0.5,0.5)$) {};
\node[small vertex, fill=red] (abc1) at ($(a) + (0.5,-0.5)$) {};
\node[small vertex, fill=green] (c1) at ($(c) + (0.25,-0.75)$) {};
\node[small vertex, fill=red] (c2) at ($(c) + (0.75,-0.25)$) {};

\draw[]{
    (a) -- (b)
    (b) -- (c)
    (a) -- (a1)
    (a) -- (ab1)
    (a) -- (abc1)
    (b) -- (ab1)
    (b) -- (abc1)
    (b) -- (b1)
    (b) -- (bc1)
    (b) -- (bc2)
    (c) -- (bc1)
    (c) -- (bc2)
    (c) -- (abc1)
    (c) -- (c1)
    (c) -- (c2)
};

% Highlight VC
\node[text node] (text2) at ($(text1.west) + (4.5,0)$) {(II) Compute Vertex Cover $C$};
\node[small vertex, double, fill=green] (a) at ($(text2.west) + (2,-2)$) {};
\node[small vertex, double, fill=green] (b) at ($(a) + (0.5,0.5)$) {};
\node[small vertex, double, fill=red] (c) at ($(a) + (1,0)$) {};
\node[small vertex, fill=blue] (a1) at ($(a) + (-0.5,-0.5)$) {};
\node[small vertex, fill=red] (ab1) at ($(a) + (-0.5,1)$) {};
\node[small vertex, fill=red] (b1) at ($(b) + (0,0.75)$) {};
\node[small vertex, fill=blue] (bc1) at ($(b) + (0.5,0.5)$) {};
\node[small vertex, fill=blue] (bc2) at ($(c) + (0.5,0.5)$) {};
\node[small vertex, fill=red] (abc1) at ($(a) + (0.5,-0.5)$) {};
\node[small vertex, fill=green] (c1) at ($(c) + (0.25,-0.75)$) {};
\node[small vertex, fill=red] (c2) at ($(c) + (0.75,-0.25)$) {};
\node[draw, dashed, rounded corners, anchor=west, label=left:{$C$}, minimum height=0.875cm, minimum width=1.5cm] (vc) at ($(a) + (-0.25,0.25)$) {};

\draw[]{
	(a) -- (b)
	(b) -- (c)
	(a) -- (a1)
	(a) -- (ab1)
	(a) -- (abc1)
	(b) -- (ab1)
	(b) -- (abc1)
	(b) -- (b1)
	(b) -- (bc1)
	(b) -- (bc2)
	(c) -- (bc1)
	(c) -- (bc2)
	(c) -- (abc1)
	(c) -- (c1)
	(c) -- (c2)
};

% Partition VC
\node[text node] (text3) at ($(text2.west) + (5.5,0)$) {(III) Guess Partition $\mathcal{P}$};
\node[small vertex, double, fill=green] (a) at ($(text3.west) + (1.5,-2)$) {};
\node[small vertex, double, fill=green] (b) at ($(a) + (0.5,0.5)$) {};
\node[small vertex, double, fill=red] (c) at ($(a) + (1,0)$) {};
\node[small vertex, fill=blue] (a1) at ($(a) + (-0.5,-0.5)$) {};
\node[small vertex, fill=red] (ab1) at ($(a) + (-0.5,1)$) {};
\node[small vertex, fill=red] (b1) at ($(b) + (0,0.75)$) {};
\node[small vertex, fill=blue] (bc1) at ($(b) + (0.5,0.5)$) {};
\node[small vertex, fill=blue] (bc2) at ($(c) + (0.5,0.5)$) {};
\node[small vertex, fill=red] (abc1) at ($(a) + (0.5,-0.5)$) {};
\node[small vertex, fill=green] (c1) at ($(c) + (0.25,-0.75)$) {};
\node[small vertex, fill=red] (c2) at ($(c) + (0.75,-0.25)$) {};
\begin{scope}[on background layer]
	\node[draw=orange, fill=orange!20, rectangle, rounded corners, anchor=west, label=left:{\scriptsize $P_1$}, minimum height=0.4cm, minimum width=1.5cm] (P1) at ($(a) + (-0.25,0)$) {};
	\node[draw=purple, fill=purple!20, rectangle, rounded corners, label=left:{\scriptsize $P_2$}, minimum height=0.4cm, minimum width=0.4cm] (P2) at ($(b) + (0,0)$) {};
\end{scope}

\draw[]{
	(a) -- (b)
	(b) -- (c)
	(a) -- (a1)
	(a) -- (ab1)
	(a) -- (abc1)
	(b) -- (ab1)
	(b) -- (abc1)
	(b) -- (b1)
	(b) -- (bc1)
	(b) -- (bc2)
	(c) -- (bc1)
	(c) -- (bc2)
	(c) -- (abc1)
	(c) -- (c1)
	(c) -- (c2)
};

% Assignment instance
\node[text node] (text4) at ($(text1.west) + (3,-4)$) {(IV) Build and Solve Assignment Instance};
\node[mwpm node, draw=orange, thick] (P1) at ($(text4.west) + (1,-1)$) {$P_1$};
\node[mwpm node, draw=purple, thick] (P2) at ($(P1) + (0,-2.5)$) {$P_2$};
\node[mwpm node, draw=black] (dummy) at ($(P2) + (0,-2.5)$) {$\ast$};
\node[mwpm node, fill=red] (r) at ($(P1) + (6,0)$) {\red};
\node[mwpm node, fill=green] (g) at ($(P2) + (6,0)$) {\green};
\node[mwpm node, fill=blue, text=white] (b) at ($(dummy) + (6,0)$) {\blue};

\draw{
    (P1) edge["$w(P_1{,}\red)=4$" pos=0.5, very thick] (r)
    (P1) edge["$w(P_1{,}\green)=2$" pos=0.2] (g)
    (P1) edge["$w(P_1{,}\blue)=5$" pos=0.8] (b)
    (P2) edge["$w(P_2{,}\red)=4$" pos=0.2] (r)
    (P2) edge["$w(P_2{,}\green)=0$" pos=0.8, very thick] (g)
    (P2) edge["$w(P_2{,}\blue)=3$" pos=0.2] (b)
    (dummy) edge["0" pos=0.8] (r)
    (dummy) edge["0" pos=0.2] (g)
    (dummy) edge["0" pos=0.5, very thick] (b)
};

\end{tikzpicture}

%% file: img/slide-bw.tex
\begin{tikzpicture}
        
% Main Graph
\node[anchor=west, draw, rectangle, rounded corners, fill=gray!20, minimum width=9cm, minimum height=1.2cm, label=left:{$G$}] (G) at (-0.5,-0.2) {};
\node[small vertex, fill=red] (a) at (0,0) {};
\node[small vertex, fill=red] (b) at (2,0) {};
\node[small vertex, fill=green] (c) at (4,0) {};
\node[small vertex, fill=green] (d) at (6,0) {};
\node[small vertex, fill=blue] (e) at (8,0) {};
\foreach \x / \y in {a/1,b/2,c/3,d/4,e/5}{
    \node[draw, rectangle, rounded corners, rotate=45, fill=brown!20, minimum size=1cm, label=above:{$C_\y$}] (C\y) at ($(\x) + (0,1.5)$) {};
}

% Cliques
\foreach \x in {a,b,c,d,e}{
    \node[small vertex, fill=red] (\x1) at ($(\x) + (-0.5,1.5)$) {};
    \node[small vertex, fill=green] (\x2) at ($(\x) + (0,2)$) {};
    \node[small vertex, fill=yellow] (\x3) at ($(\x) + (0,1)$) {};
    \node[small vertex, fill=blue] (\x4) at ($(\x) + (0.5,1.5)$) {};
}

\draw[thick]{
    (a) -- (b)
    (a) edge[bend right] (c)
    (b) -- (c)
    (d) -- (e)
    (a3) -- (b1)
    (b3) -- (c1)
    (c3) -- (d1)
    (d3) -- (e1)
    (c) -- (c3)
    (e) -- (e3)
};

\foreach \x in {a,b,c,d,e}
    \foreach \y in {1,2,3,4}
        \foreach \z in {4,3,2,1}{
            \draw[thick] (\x\y) -- (\x\z);
            \ifnum \y = \z-1
                \breakforeach
            \fi
        }

\end{tikzpicture}

%% file: img/flip-dtc.tex
\begin{tikzpicture}

% Complete Subgraph
\node[draw, ellipse, label={[shift={(-4.5,-0.5)}]{\large $C$}}, minimum width=13.5cm, minimum height=4cm, anchor=west] (C) at (0,0) {};
\node[draw, ellipse, fill=gray!20, label={[shift={(0.75,-0.75)}]{$R_1$}}, minimum width=2.5cm, minimum height=1.5cm, anchor=west] (R1) at ($(C.west) + (0.5,0)$) {};
\node[draw, ellipse, fill=gray!20,, label={[shift={(0.75,-0.75)}]:{$R_2$}}, minimum width=2.5cm, minimum height=1.5cm, anchor=west] (R2) at ($(C.west) + (3,-1)$) {};
\node[draw, ellipse, fill=gray!20,, label={[shift={(0.75,-0.75)}]:{$R_3$}}, minimum width=2.5cm, minimum height=1.5cm, anchor=west] (R3) at ($(C.west) + (5.5,1)$) {};
\node[draw, ellipse, fill=gray!20, label={[shift={(-0.75,-0.75)}]{$R_i$}}, minimum width=2.5cm, minimum height=1.5cm, anchor=west] (R4) at ($(C.west) + (10.5,0)$) {};

\node[small vertex, fill] (1a) at ($(R1.west) + (0.5,0)$) {};
\node[small vertex, fill] (1b) at ($(1a) + (0.5,-0.35)$) {};
\node[small vertex, fill] (1c) at ($(1a) + (0.55,0.25)$) {};

\node[small vertex, fill] (2a) at ($(R2.west) + (0.5,0)$) {};
\node[small vertex, fill] (2b) at ($(2a) + (0.4,-0.4)$) {};
\node[small vertex, fill] (2c) at ($(2a) + (0.4,0.4)$) {};
\node[small vertex, fill] (2d) at ($(2a) + (0.8,0)$) {};

\node[small vertex, fill] (3a) at ($(R3.west) + (0.5,0)$) {};
\node[small vertex, fill] (3b) at ($(3a) + (0.5,-0.35)$) {};
\node[small vertex, fill] (3c) at ($(3a) + (0.55,0.25)$) {};

\node[small vertex, fill] (4a) at ($(R4.east) + (-0.5,0)$) {};
\node[small vertex, fill] (4b) at ($(4a) + (-0.5,0)$) {};

\foreach \x in {1,2,3,4}{
    \draw (\x a) -- (\x b);
}
\foreach \x in {1,2,3}
    \foreach \a in {a,b,c}{
        \draw (\x \a) -- (\x c);
}

\foreach \a in {a,b,c}{
    \draw (2\a) -- (2d);
}

\draw{
    (R1) -- (R2)
    (R1) -- (R3)
    (R1) edge[bend right=1.5mm] (R4)
    (R2) -- (R3)
    (R2) edge[bend right=2mm] (R4)
    (R3) -- (R4)
};

% Dots
\draw[line width=1mm, line cap=round, dash pattern=on 0pt off 2.5mm] ($(R4.west) + (-1.5,0)$) -- ($(R4.west) + (-2,0)$);

% Partition Blocks
\node[draw, rounded corners, label=left:{\large $S$}, minimum width=7cm, minimum height=2cm, anchor=west] (S) at ($(C.west) + (1.5,-6)$) {};
\node[draw, rounded corners, fill=red!20, label={[shift={(0,0.6)}]below:{$P_1$}}, minimum width=2cm, minimum height=1.5cm, anchor=north west] (P1) at ($(S.north west) + (0.25,-0.25)$) {};
\node[draw, rounded corners, fill=blue!20, label={[shift={(0,0.6)}]below:{$P_j$}}, minimum width=2cm, minimum height=1.5cm, anchor=west] (P2) at ($(P1.east) + (2.5,0)$) {};

\node[small vertex, fill] (1a) at ($(P1.north west) + (0.75,-0.5)$) {};
\node[small vertex, fill] (1b) at ($(1a) + (0.5,-0.35)$) {};
\node[small vertex, fill] (1c) at ($(1a) + (0.55,0.25)$) {};

\node[small vertex, fill] (1a) at ($(P2.north west) + (0.75,-0.5)$) {};
\node[small vertex, fill] (1b) at ($(1a) + (0.5,-0.35)$) {};

\draw[dashed]{
    ($(P1.south east) + (0,0.25)$) edge[thick,dotted] ($(P2.south west) + (0,0.25)$)
    (P1) -- (R1)
    (P1) -- (R2)
    (P2) -- (R2)
    (P2) -- (R3)
};

% Dots
\draw[line width=1mm, line cap=round, dash pattern=on 0pt off 2.5mm] ($(P1.east) + (1,0)$) -- ($(P2.west) + (-1,0)$);

% No Adjacency Class
\node[draw=red, circle, line width=2mm, minimum size=1cm] (em) at ($(S.east) + (3,0)$) {$\bot$};

\draw[ultra thick, gray, -Stealth]{
	(P2) edge[dashed] (R4)
	(P2) edge (em)
	(P2) edge[dashed] (R1)
	(P1) edge[bend right, dashed] (em)
	(P1) edge[bend right=0.5cm, dashed] (R3)
	(P1) edge[bend right=0.3cm] (R4)
};

\end{tikzpicture}

%% file: img/flip-hard.tex
\begin{tikzpicture}

\node (text1) at (2,-1) {\textsc{List Coloring}};
\node (text2) at ($(text1) + (6,0)$) {\textsc{Coloring Discovery}};

\node[vertex, minimum size=4mm, label={$\{\red,\blue\}$}] (u) at (0,2) {$u$};
\node[vertex, minimum size=4mm, label={$\{\red,\green,\blue\}$}] (v) at (2,2) {$v$};
\node[vertex, minimum size=4mm, label={$\{\green\}$}] (w) at (4,2) {$w$};
\node[vertex, minimum size=4mm, label=below:{$\{\green,\blue\}$}] (x) at (4,0) {$x$};

\draw[thick]{
    (u) -- (v)
    (v) -- (w)
    (v) -- (x)
    (w) -- (x)
};

\node[small vertex, fill=red, label={[label distance=-0.5mm]below:{\scriptsize $u$}}] (u) at ($(u) + (6,0)$) {};
\node[small vertex, fill=red, label={[label distance=-0.5mm]:{\scriptsize $v$}}] (v) at ($(v) + (6,0)$) {};
\node[small vertex, fill=red, label={[label distance=-0.5mm]below left:{\scriptsize $w$}}] (w) at ($(w) + (6,0)$) {};
\node[small vertex, fill=red, label={[label distance=-0.5mm]below:{\scriptsize $x$}}] (x) at ($(x) + (6,0)$) {};

\node[small vertex, fill=green, label={[label distance=-1mm]:{\scriptsize $u_1^\green$}}] (u1g) at ($(u) + (-0.5,1)$) {};
\node[small vertex, fill=green, label={[label distance=-1mm]:{\scriptsize $u_{5}^\green$}}] (ubg) at ($(u) + (0.5,1)$) {};
\node[small vertex, fill=red, label={[label distance=-1mm]:{\scriptsize $w_1^\red$}}] (w1r) at ($(w) + (-0.5,1)$) {};
\node[small vertex, fill=red, label={[label distance=-1mm]:{\scriptsize $w_{5}^\red$}}] (wbr) at ($(w) + (0.5,1)$) {};
\node[small vertex, fill=blue, label={[label distance=-1mm]right:{\scriptsize $w_1^\blue$}}] (w1b) at ($(w) + (1,0.5)$) {};
\node[small vertex, fill=blue, label={[label distance=-1mm]right:{\scriptsize $w_{5}^\blue$}}] (wbb) at ($(w) + (1,-0.5)$) {};
\node[small vertex, fill=red, label={[label distance=-1mm]right:{\scriptsize $x_1^\red$}}] (x1r) at ($(x) + (1,0.5)$) {};
\node[small vertex, fill=red, label={[label distance=-1mm]right:{\scriptsize $x_5^\red$}}] (xbr) at ($(x) + (1,-0.5)$) {};

\draw[thick]{
    (u) -- (v)
    (v) -- (w)
    (v) -- (x)
    (w) -- (x)
    
};

\draw[]{
	(u) -- (u1g)
	(u) -- (ubg)
	(w) -- (w1r)
	(w) -- (wbr)
	(w) -- (w1b)
	(w) -- (wbb)
	(x) -- (x1r)
	(x) -- (xbr)
    (u) -- ($(u) + (-0.25, 0.9)$)
    (u) -- ($(u) + (0, 0.9)$)
    (u) -- ($(u) + (0.25, 0.9)$)
    (w) -- ($(w) + (0.9, 0.25)$)
    (w) -- ($(w) + (0.9,0)$)
    (w) -- ($(w) + (0.9,-0.25)$)
    (w) -- ($(w) + (-0.25, 0.9)$)
    (w) -- ($(w) + (0, 0.9)$)
    (w) -- ($(w) + (0.25, 0.9)$)
    (x) -- ($(x) + (0.9, 0.25)$)
    (x) -- ($(x) + (0.9,0)$)
    (x) -- ($(x) + (0.9,-0.25)$)
};

\draw[thick, dotted]{
    (u1g) -- (ubg)
    (w1r) -- (wbr)
    (w1b) -- (wbb)
    (x1r) -- (xbr)
};

\end{tikzpicture}